\documentclass[conference]{IEEEtran} 

\usepackage{amssymb}
\usepackage{amsmath}
\usepackage{amsmath,bm}
\usepackage{amsthm}
\usepackage{array}
\usepackage{graphicx}
\usepackage{subfigure}
\usepackage{cite}
\usepackage{enumerate}
\usepackage{enumitem}
\usepackage{algorithm}
\usepackage{algorithmic}
\usepackage{color, soul}
\usepackage[labelsep=period]{caption}
\usepackage{tablefootnote}
\usepackage{lipsum}
\allowdisplaybreaks[4]
\begin{document}
\newtheorem{theorem}{\bf~~Theorem}
\newtheorem{remark}{\bf~~Remark}
\newtheorem{definition}{\bf~~Definition}
\newtheorem{lemma}{\bf~~Lemma}
\newtheorem{preliminary}{\bf~~Preliminary}
\renewcommand\arraystretch{0.9}
\newtheorem{prop}{Proposition}

\title{Near-Far Field Channel Modeling for Holographic MIMO Using Expectation-Maximization Methods\vspace{-0.3cm}}

\author{\IEEEauthorblockN{
        {Houfeng Chen\IEEEauthorrefmark{1}}, 
        {Shuhao Zeng\IEEEauthorrefmark{2}}, 
        {Hao Guo\IEEEauthorrefmark{3}}, 
        {Tommy Svensson\IEEEauthorrefmark{3}}, and 
        {Hongliang Zhang\IEEEauthorrefmark{2}}
        }
    \IEEEauthorblockA{\IEEEauthorrefmark{1}School of Electronic Engineering and Computer Science, Peking University, Beijing, China}
    \IEEEauthorblockA{\IEEEauthorrefmark{2}Department of Electronics, Peking University, Beijing, China}
    \IEEEauthorblockA{\IEEEauthorrefmark{3}Department of Electrical Engineering, Chalmers University of Technology, Gothenburg, Sweden}
    \vspace{-0.9cm}
}

\maketitle
\begin{abstract}
Holographic Multiple-Input Multiple-Output (HMIMO), which densely integrates numerous antennas into a limited space, is anticipated to provide higher rates for future 6G wireless communications. The increase in antenna aperture size makes the near-field region enlarge, causing some users to be located in the near-field region. Thus, we are facing a hybrid near-field and far-field communication problem, where conventional far-field modeling methods may not work well. 
In this paper, we propose a near-far field channel model that does not presuppose whether each path is near-field or far-field, different from the existing work requiring the ratio of the number of near-field paths to that of far-field paths as prior knowledge. However, this gives rise to a new challenge for accurately modeling the channel, as conventional methods of obtaining channel model parameters are not applicable to this model. Therefore, we propose a new method, Expectation-Maximization (EM)-based Near-Far Field Channel Modeling, to obtain channel model parameters, which considers whether each path is near-field or far-field as a hidden variable, and optimizes the hidden variables and channel model parameters through an alternating iteration method. Simulation results show that our method is superior to conventional near-field and far-field algorithms in fitting the near-far field channel in terms of outage probability.


\end{abstract}
\vspace{0.2cm}

\vspace{0.1cm}
\section{Introduction}
Holographic Multiple-Input-Multiple-Output (HMIMO) is considered as a promising wireless communication technology for 6G networks~\cite{HMIMO_review}. By densely integrating numerous antennas or reconfigurable elements within a limited space~\cite{HMIMO_integrate}, HMIMO can significantly enhance beamforming gains and achieve higher spectral efficiency. Due to the increased radiation aperture size of the antenna array for HMIMO, the radiating near-field region of the antenna array has been extended significantly. Consequently, some scatters are located in the antenna array's near-field region, where electromagnetic waves exhibit spherical wave characteristics, while the others are situated in the far-field region, where electromagnetic waves are characterized by planar waves. This phenomenon is known as near-far field communication~\cite{HMIMO_integrate}.

Most existing works focus on either the far-field~\cite{OMP_2016} or the near-field~\cite{Cui_2022} channel modeling. In \cite{OMP_2016}, the proposed far-field channel relies on the assumption of far-field planar waves, based on which, the array steering vector of the far-field channel is only related to the angles between the base station (BS) and the scatters. Differently, in \cite{Cui_2022}, the spherical assumption makes the array steering vector of the near-field channel depend on both the angles and the distances between the BS and the scatters. Few initial works \cite{Wei_2022}, \cite{Hu_2023} consider near-far field channel models that contain both far-field and near-field components, which requires the ratio of the number of near-field paths to that of far-field paths as prior knowledge.
However, due to the complex wireless propagation environment, it is challenging to acquire this ratio in practical systems.

To address the above issue, we propose a near-far field channel model without requiring the ratio of near-field paths to far-field paths as prior knowledge. In the designed channel model, it is not pre-known whether each path is near-field or far-field, which, however, gives rise to a new challenge: the specific expression of the array steering vector of each path is also unknown. This leads to performance degradation using expressions of existing channel models, which requires a new model to capture wireless propagation characteristics.

In this paper, we propose a near-far field channel model to address the above issue. Specifically, note that whether each path is near-field or far-field is unknown and affects the specific expression of the channel distribution, which implies that it is a hidden variable. Therefore, we propose an Expectation-Maximization (EM)-based Near-Far Field Channel Modeling algorithm. This algorithm iteratively alternates between two key steps: the expectation step (E-step) and the maximization step (M-step), thereby optimizing both hidden variables and channel model parameters. Moreover, the outage probability using the proposed channel model is derived, based on which the performance of the proposed method can be analytically evaluated. Simulation results show that our method is superior to conventional near-field and far-field algorithms in fitting the near-far field channel.



The rest of the paper is organized as follows. In Section \ref{sysmod}, the communication scenario, channel model and signal model are introduced. In Section \ref{section3}, we first build the near-far field channel and propose the EM algorithm for near-far field channel modeling. Then, the outage probability is derived for evaluating the system's performance. Simulation results are provided in Section \ref{SIM} to show the effectiveness of the proposed near-far channel model. Finally, conclusion remarks are drawn in Section \ref{Conclusion}.

\begin{figure}
    \centering
    \includegraphics[width=0.42\textwidth]{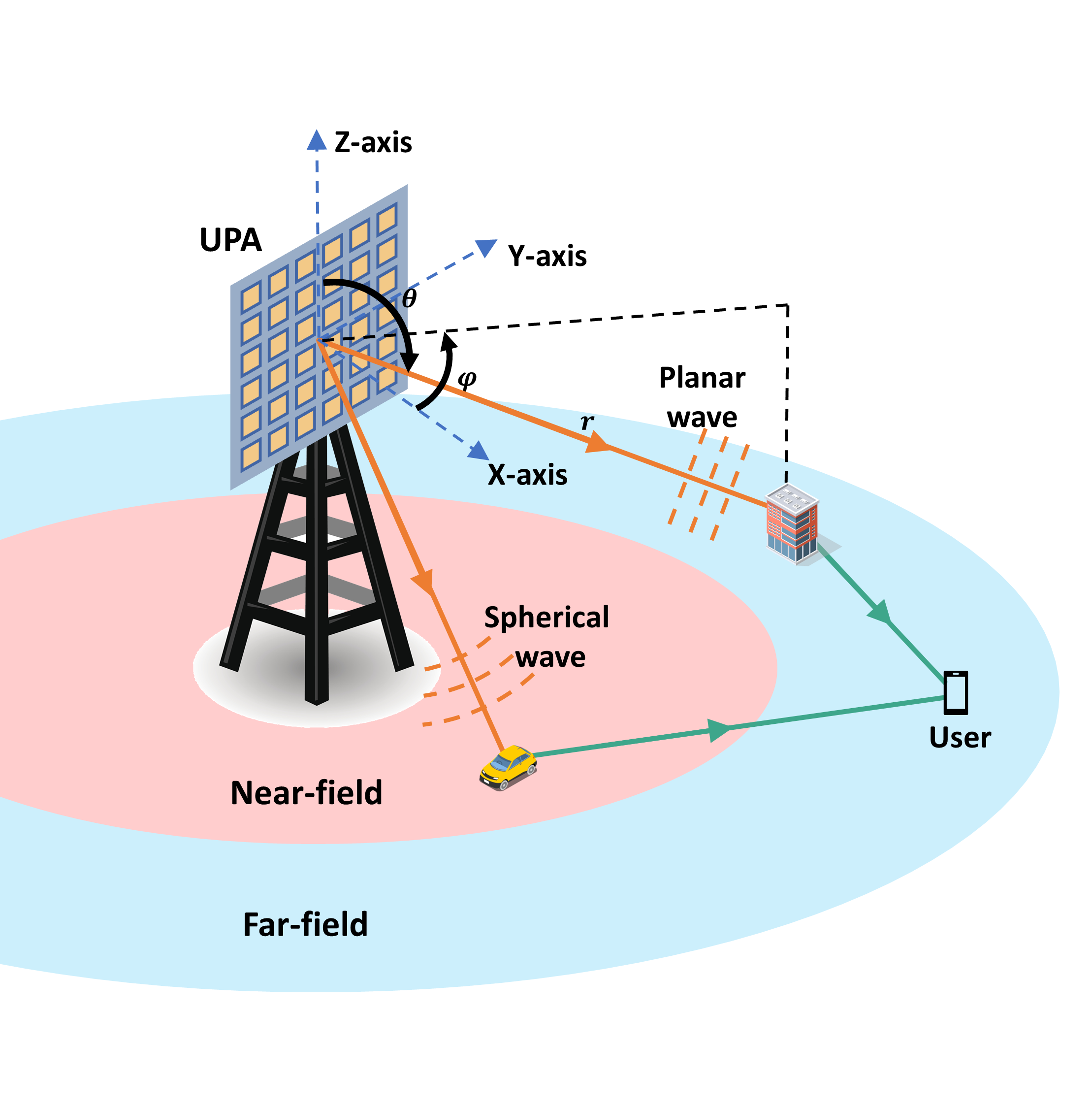}
    \caption{The HMIMO system equipped with UPA when scatters in both near field and far field. }
    \label{figure_scenario}
    \vspace{-0.2cm}
\end{figure}

\vspace{-0.1cm}
\section{System Model}\label{sysmod}
In this section, we first introduce the communication scenario of holographic MIMO systems. Then, we introduce the near-far field channel model. Finally, the signal model in this communication scenario is given.

\subsection{Scenario Description} 
We consider a downlink HMIMO communication system, where the BS is equipped with one extremely large uniform planar array (UPA) consisting of $N = N_1 \times N_2$ antennas to communicate with single-antenna users.
As depicted in Fig.~\ref{figure_scenario}, the electromagnetic radiation field of the antenna array at the BS can be divided into two regions: near-field and far-field~\cite{Cui_2022}. The boundary of these two fields is typically determined by the Rayleigh distance~\cite{Rayleigh Distance}, which increases with the size of the antenna array. Since the antenna array under consideration contains a large number of antenna elements, the boundary is comparable to the cell radius. Consequently, the scatters can be located in both the near-field and far-field regions~\cite{Wei_2022}.

\subsection{Channel Models}
We consider the multipath channel from the BS to users as a superposition of $L$ deterministic paths and additional diffuse signal components \cite{J.D._2020}. These deterministic paths represent the primary contribution of radio waves propagating from the BS to users through scatters, while the diffuse signal components correspond to the other numerous but weak multipath components~\cite{channel_model_2022}, where the joint effect of these diffuse components can be modeled with a random variable $\bm{Z}$, which follows zero-mean complex Gaussian distributions, i.e., $\bm{Z} \sim \mathcal{CN}(0, \sigma^2 \bm{I})$. Therefore, the channel can be expressed as
{\setlength\abovedisplayskip{0.1cm}
        \setlength\belowdisplayskip{0.1cm}
        \begin{align}
            \label{channel model}
            \bm{h} = \sum_{l=1}^{L} \beta_l \bm{s}(\theta_l, \phi_l, r_l) + \bm{Z},
        \end{align}
    where $\beta_l$ is the path gain of the $l$-th deterministic path}, $\bm{s}(\theta_l, \phi_l, r_l)$ is the array steering vector of the UPA, and it takes on different expressions for the far-field and near-field paths, respectively, which will be elaborated on in the following.

\subsubsection{Far-field Path}
For a far-field path, where the scatter corresponding to the path is in the far field of the antenna array, the corresponding steering vector follows the planar wave assumption and can be expressed as~\cite{OMP_2016}
    {\setlength\abovedisplayskip{0.3cm}
        \setlength\belowdisplayskip{0.3cm}
        \begin{align}
        \label{far_vec}
            &\bm{s}(\theta, \phi, r) = \bm{a}_z(\theta) \otimes \bm{a}_y(\theta, \phi) \nonumber \\
            &= \frac{1}{\sqrt{N}}[1, e^{-j\frac{2\pi}{\lambda}d cos(\theta)}, \cdots, e^{-j\frac{2\pi}{\lambda}(N_2-1)d cos(\theta)}]^H \otimes \nonumber \\ 
            & [1, e^{-j\frac{2\pi}{\lambda}d sin(\phi)sin(\theta)}, \cdots, e^{-j\frac{2\pi}{\lambda}(N_1-1)d sin(\phi)sin(\theta)}]^H.
        \end{align}
 Here, $\otimes$ is the Kronecker product,} $\phi$ and $\theta$ represent angles of arrival (AoA) in azimuth and elevation as shown in Fig.~\ref{figure_scenario}. $d$ is the minimum antenna spacing of the array, and $\lambda$ is the carrier wavelength.

\subsubsection{Near-field Path}
When the scatter is located in the near field, the spherical wave assumption can more accurately capture the characteristics of near-field electromagnetic waves compared with the planar wave assumption. Therefore, the near-field steering vector can be expressed as~\cite{near_field_vector}
    {\setlength\abovedisplayskip{0.1cm}
        \setlength\belowdisplayskip{0.1cm}
        \begin{align}
            \label{near_vec}
            \bm{s}(\theta, \phi, r) = \frac{1}{\sqrt{N}}&[e^{-j\frac{2\pi}{\lambda}(r^{(1, 1)}-r)}, \cdots, e^{-j\frac{2\pi}{\lambda}(r^{(n_1, n_2)}-r)}, \nonumber \\
           &\cdots, e^{-j\frac{2\pi}{\lambda}(r^{(N_1, N_2)}-r)} ]^H,
        \end{align}
        where $\phi$, $\theta$ and $r$ represent the azimuth AoA,} elevation AoA, and the distance between the center of the array and the scatter, respectively. Furthermore, $r^{(n_1, n_2)}$ represents the distance from the $(n_1, n_2)$-th antenna of the array to the scatter, which can be expressed as
        {\setlength\abovedisplayskip{0.3cm}
        \setlength\belowdisplayskip{0.3cm}
        \begin{align}
            r^{(n_1, n_2)}  =&  \big[( r \sin \theta \cos \phi ) ^ { 2 } + ( r \sin \theta \sin \phi  - n_1 d ) ^ { 2 } \nonumber \\
            &+ ( r \cos \theta - n_{2}d)^2 \big]^{\frac{1}{2}} .
        \end{align}}
        

Based on (\ref{channel model}), the distribution of the near-far field channel $\bm{h}$ can be written as
{\setlength\abovedisplayskip{0.1cm}
        \setlength\belowdisplayskip{0.1cm}
        \begin{align}
        \label{channel distribution}
            \bm{h} \sim \mathcal{CN} \Big(\sum_{l=1}^L \beta_l \bm{s}(\theta_l, \phi_l, r_l), \sigma^2\bm{I} \Big).
        \end{align}}

\subsection{Signal Models}
Assume that the BS transmits orthogonal signals to different users, ensuring that the signals received by different users are interference-free. Therefore, signals received by an arbitrary user can be denoted as
    {\setlength\abovedisplayskip{0.1cm}
        \setlength\belowdisplayskip{0.1cm}
        \begin{align}
            y = \sqrt{P_{\rm T}} \bm{h}^H \bm{x} + n,
        \end{align}
        where $P_{\rm T}$ is the transmit power of BS,} $\bm{h}$ is the channel in (\ref{channel model}), $y$ is the received complex baseband signal, $\bm{x}\sim \mathbb{C}^{N\times 1}$ is the transmitted signal with $\mathbb{E}[{||x||^2}]=1$, and $n \sim \mathcal{CN}(0, \hat{\sigma}^2)$ is the Gaussian noise.
        
\vspace{0.1cm}
\section{EM-Based Near-Far field Channel Modeling and Outage Probability Performance} \label{section3}
In this section, we aim to extract the channel parameters values $\Theta = \{\beta_l, \theta_l, \phi_l, r_l\}$ from the observed channel samples, so as to validate the proposed channel model and facilitate the analysis of system performance.

\vspace{-0.1cm}
\subsection{EM-Based Near-Far field Channel Modeling}
\vspace{-0.1cm}
Use $\bm{z} = [z_1, \cdots, z_L]^T$ to indicate whether the $L$ deterministic paths are near-field or far-field paths, where $z_l$ is a binary variable, with $z_l=0$ and $z_l=1$ indicating that the $l$-th path is near-field and far-field, respectively. Recall that the value of $z_l$ determines the representation of the steering vector $\bm{s}(\theta_l, \phi_l, r_l)$, as indicated in (\ref{far_vec}) and (\ref{near_vec}). Therefore, $z_l$ also has an effect on the distribution of $\bm{h}$ according to (\ref{channel distribution}).
As a result, to apply the traditional Maximum Likelihood Estimation (MLE) method to estimate the channel parameters, the value of $\bm{z}$ needs to be known in advance\footnote{If $\bm{z}$ is pre-known, the distribution of the channel can be derived from (\ref{channel distribution}), and the channel parameters $\Theta$ can be obtained based on the MLE method~\cite{Cui_2022}.}. However, $\bm{z}$ is difficult to obtain in practice as it's a hidden variable. To address this issue, we propose a near-far field channel modeling scheme based on the EM algorithm.
The main idea of the channel modeling scheme is to iteratively apply two steps, namely, the expectation (E)-step and the maximization (M)-step, where the posterior distribution of $\bm{z}$ and $\Theta$ are optimized, respectively. In the following, we will describe in detail how to perform the EM algorithm for near-far field channel modeling:

\subsubsection{\textbf{E-step}} 
In this step, given $S$ independent and identically distributed (i.i.d.) channel samples $\bm{H} = [\bm{h}_1, \cdots, \bm{h}_S]$ and the current estimates of $\Theta$ (denoted as $\Theta^{j-1}$), the posterior distribution of $\bm{z}$, i.e., $p(\bm{z}|\bm{H}, \Theta^{j-1})$ can be acquired according to Bayes' theorem as
{\setlength\abovedisplayskip{0.1cm}
    \setlength\belowdisplayskip{0.1cm}
\begin{align}
    \label{posterior_distribution}
    p(\bm{z}|\bm{H}, \Theta^{j-1}) &= \frac{p(\bm{H}, \bm{z} | \Theta^{j-1})}{\sum_{\bm{z'}} p(\bm{H}, \bm{z'} | \Theta^{j-1})} \nonumber \\
    &= \frac{p(\bm{z}|\Theta^{j-1}) p(\bm{H} | \bm{z}, \Theta^{j-1})}{\sum_{\bm{z'}} p(\bm{z'}|\Theta^{j-1}) p(\bm{H} | \bm{z'}, \Theta^{j-1})} ~,
\end{align}
    where $p(\bm{H}, \bm{z} | \Theta ^{j-1})$ is the joint probability of $\bm{H}$ and $\bm{z}$.} According to the conditional probability and the assumption of i.i.d. between channel samples,  $p(\bm{H}, \bm{z}|\Theta^{j-1})$ can be further expressed as
     {\setlength\abovedisplayskip{0.1cm}
            \setlength\belowdisplayskip{0.1cm}
            \begin{align}
            \label{conditional_H_z_to_Theta}
                p(\bm{H}, \bm{z}|\Theta^{j-1}) &= p(\bm{H} | \bm{z}, \Theta^{j-1}) \cdot p(\bm{z} | \Theta^{j-1}) \nonumber \\
                &= \Big(\prod_{s=1}^S p(\bm{h}_s| \bm{z}, \Theta^{j-1}) \Big)  \cdot p(\bm{z}|\Theta^{j-1}) ~.
            \end{align}
    Here, we would like to point out that the distribution of channel $\bm{h}_s$ for given $\bm{z}$ and $\Theta^{j-1}$ follows complex Gaussian distributions $\mathcal{CN}\Big(\sum_{l=1}^L \beta_l^{j-1} \bm{s}^H(\theta_l^{j-1}, \phi_l^{j-1}, r_l^{j-1} | z_l), \sigma^2\bm{I}\Big)$ and the probability density function can be calculated as}
    {\setlength\abovedisplayskip{0.1cm}
        \setlength\belowdisplayskip{0.1cm}
        \begin{align}
        \label{conditional_distribution_h_to_z_Theta}
            p(\bm{h}_s | \bm{z}, \Theta^{j-1})=\frac{1}{(\pi)^{N} |\sigma^2 I|} e^{-\bm{r}^{H}(\sigma^2\bm{I})^{-1}\bm{r}},
        \end{align}
        where $\bm{r} = \bm{h}_s - \sum\limits_{l=1}^L \beta_l \bm{s}^H(\theta_l^{j-1}, \phi_l^{j-1}, r_l^{j-1} | z_l)$. Depending on the value of $z_l$, $\bm{s}(\theta_l^{j-1}, \phi_l^{j-1}, r_l^{j-1}|z_l)$ can be determined as either the far-field steering vector defined in (\ref{far_vec}) or the near-field steering vector defined in (\ref{near_vec}). }
        
\setlength{\textfloatsep}{0.1cm} 
\begin{algorithm}[t]
    \caption{Expectation-Maximization Algorithm in Near-Far Field Channel Modeling}
    \renewcommand{\algorithmicrequire}{\textbf{Input:}}
    \renewcommand{\algorithmicensure}{\textbf{Output:}}
    \label{algorithm1}
    \begin{algorithmic}[1]
        \REQUIRE $S$ channel samples of $\bm{H}$, number of paths $L$, the convergence threshold $\delta$
        \STATE \textbf{Initialization}: $\Theta^0$
        \WHILE{True}
        \STATE \textit{\textbf{E-step}}: \\
        \quad Calculate the conditional expectation of the log-likelihood based on (\ref{posterior_distribution})-(\ref{conditional_distribution_h_to_z_Theta})
        \vspace{0.1cm}
        \STATE \textit{\textbf{M-step}}: 
        \STATE  \quad update model parameter $\Theta$ based on (\ref{conditional_expectation_of_loglikelihood})-(\ref{gradient_descent})
        \IF {$\|\Theta_{j} - \Theta_{j-1}\| \le \delta$}
        \STATE break
        \ENDIF
        \ENDWHILE
        \ENSURE Estimated channel model parameters $\Theta$,  posterior probability distribution of $\bm{z}$ 
    \end{algorithmic}
\end{algorithm}

\subsubsection{\textbf{M-step}}
In this step, given the posterior distribution of $\bm{z}$ in (\ref{posterior_distribution}), the model parameter $\Theta$ is updated using the MLE method. To be specific, based on Jensen's inequality and its equivalence conditions, the log-likelihood function is simplified to its conditional expectation~\cite{EM_algorithm_paper}, which can be formulated as
{\setlength\abovedisplayskip{0.1cm}
    \setlength\belowdisplayskip{0.1cm}
\begin{align}
\label{conditional_expectation_of_loglikelihood}
    &Q(\Theta | \Theta^{j-1}) = \mathbb{E}_{\bm{z}\sim p(\bm{z}|\bm{H}, \Theta^{j-1})} \Big[\log( p(\bm{H}, \bm{z} | \Theta) )\Big]~,
\end{align}
where $\mathbb{E}_{\bm{z}\sim p(\bm{z}|\bm{H}, \Theta^{j-1})}[\cdot]$ is the expectation with respect to $\bm{z}$ conditioned on $\bm{H}$ and $\Theta^{j-1}$.} Then, based on the MLE method, the estimated value of $\Theta$ in the $j$-th iteration can be obtained by maximizing the likelihood function in (\ref{conditional_expectation_of_loglikelihood}), i.e., the model parameter $\Theta$ in the $j$-th iteration can be updated as
{\setlength\abovedisplayskip{0.1cm}
    \setlength\belowdisplayskip{0.1cm}
    \begin{align}
    \label{update_parameter}
        \Theta^{j} = \mathop{\arg\max}\limits_{\Theta} Q(\Theta | \Theta^{j-1}).
    \end{align}
Here, we use gradient descent to update $\Theta$.} In the $j$-th iteration, the $n$-th gradient descent of $\Theta$ can be expressed as
{\setlength\abovedisplayskip{0.1cm}
    \setlength\belowdisplayskip{0.1cm}
    \begin{align}
    \label{gradient_descent}
        \Theta^{j,n} = \Theta^{j, n-1} - l_1 \cdot \nabla_\Theta Q(\Theta|\Theta^{j, n-1}),
\end{align}}where $\Theta^{j,n}$ denotes the estimated value of $\Theta$ after the $n$-th gradient descent in the $j$-th iteration, and $l_1$ denotes the learning rate, which is chosen by Backtracking-Armijo Line Search~\cite{Armijo}, ensuring that the iteration results converge. The gradient $\nabla_\Theta Q(\Theta|\Theta^{j-1})$ can be obtained by stacking all gradients of $\theta_l$, $\phi_l$, $r_l$ and $\beta_l$ in a column vector, where the gradients of $\theta_l$, $\phi_l$ and $r_l$ can be calculated as

{\setlength\abovedisplayskip{0.0cm}
    \setlength\belowdisplayskip{0.0cm}
    \small
    \begin{align}
       \frac{\partial Q(\Theta|\Theta^{j\!-\!1})}{\partial (\cdot)}\!=\!\mathbb{E}_{\bm{z}\sim p(\bm{z}|\bm{H}, \Theta^{j \! - \! 1})} \!\bigg\{\!\frac{2}{\sigma^2} \!\sum_{s=1}^S \! \Re \!\Big\{\beta_l' \!\cdot\! \frac{\partial \bm{s}^H\!(\theta_l, \phi_l, r_l | z_l)}{\partial (\cdot)} \!\cdot\!\bm{r}\!\Big\} \!\bigg\}.
\end{align}}Here $(\cdot)$ can be $\theta_l$, $\phi$ and $r_l$. Since $\beta_l$ is a complex value, the gradient of $\beta_l$ has a different form, which can be calculated as

{\setlength\abovedisplayskip{-0.1cm}
    \setlength\belowdisplayskip{0.0cm}
    \small
    \begin{align}
       \frac{\partial Q(\Theta|\Theta^{j-1})}{\partial \beta_l} \!=\! \mathbb{E}_{\bm{z}\sim p(\bm{z}|\bm{H}, \Theta^{j\!-\!1})} \! \bigg\{\! \frac{1}{\sigma^2} \sum_{s=1}^S \! \big(\bm{r}^H \!\cdot\! \bm{s}(\theta_l, \phi_l, r_l | z_l) \big) \! \bigg\}.
\end{align}}The derivation of the gradient $\nabla_\Theta Q(\Theta|\Theta^{j-1})$ is given in \textbf{Appendix \ref{appendix1}}.

Repeat the E-step and the M-step, when  $||\Theta^{j}-\Theta^{j-1}||$ is less than the preset convergence threshold $\delta$, the algorithm ends and we can obtain the estimated value of $\Theta$ as well as the hidden variable $\bm{z}$. The proposed EM-based Near-Far Field Channel Modeling algorithm is summarized in \textbf{Algorithm~\ref{algorithm1}}.

\subsection{Outage Probability Performance}
Here, the outage probability (OP) is considered as a metric in this paper for evaluating the system's performance. The expression of OP for the channel model is given in the following proposition.
\begin{prop} \label{proposition1}
The expression of OP for the channel model can be calculated as
{\setlength\abovedisplayskip{0.1cm}
    \setlength\belowdisplayskip{0cm}
\begin{align}
    \label{OP_expression}
    OP &= Pr\Big\{ \log_2(1+\frac{P_{\rm T}}{\hat{\sigma}^2}\bm{h}^H\bm{h}) < R_{\rm th} \Big\} \nonumber \\ 
    &= \! F\bigg(\!\frac{(2^{R_{\rm th}} - 1)\hat{\sigma}^2 }{P_{\rm T} \cdot \sigma^2/2}; 2N, \frac{\left\Vert\sum_{l=1}^L \beta_l \bm{s}(\theta_l, \phi_l, r_l)\right\Vert^2}{\sigma^2/2} \!\bigg),
\end{align}
where $R_{\rm th}$ denotes the target rate}, and $F(x; k, \lambda)$ denotes the cumulative distribution function (CDF) of the non-central chi-squared distribution $\chi_k'^2(\lambda)$. It can be expressed as
{\setlength\abovedisplayskip{0.1cm}
\setlength\belowdisplayskip{0.1cm}
\begin{align}
    F(x; k, \lambda) = e^{\lambda/2} \sum_{t=0}^{\infty} \frac{(\lambda/2)^t}{t!} \frac{\gamma(\frac{k+2j}{2}, \frac{x}{2})}{\Gamma(\frac{k+2j}{2})},
\end{align}
where $\gamma(k,z)$ is the lower incomplete gamma function,} and $\Gamma(s)$ is the ordinary gamma function.
\end{prop}

\begin{proof}
\vspace{-0.2cm}
See \textbf{Appendix \ref{appendix2}}.
\vspace{-0.2cm}
\end{proof}
Based on $\bm{z}$ and parameters $\Theta$, the OP can be obtained from (\ref{OP_expression}).

\section{Simulation Results} \label{SIM}

In this section, simulation results are provided to verify the effectiveness of the proposed channel modeling strategy. We consider a downlink HMIMO communication system, and the simulation configurations are shown in Tabel \ref{sim_config}. $S=100$ channel samples are generated as the input of the channel modeling process of \textbf{Algorithm \ref{algorithm1}}. For comparison, the following schemes are introduced as benchmarks.

\begin{itemize}[leftmargin=*]
    
    \item \textbf{Far-field Scheme}~\cite{OMP_2016}: The algorithm in \cite{OMP_2016} models the far-field channel and uses the Orthogonal Matching Pursuit (OMP) algorithm to acquire channel parameters.
    \item \textbf{Near-field Scheme}~\cite{Cui_2022}: The algorithm in ~\cite{Cui_2022} can be used for near-field channel modeling,  and the polar-domain simultaneous orthogonal matching pursuit (P-SOMP) algorithm based on the OMP algorithm is used to obtain near-field channel parameters.
    \item \textbf{Monte Carlo (MC) Simulations}: $10^5$ channel samples are generated to validate the effectiveness of the modeled channel from the proposed EM framework, i.e., \textbf{Algorithm~\ref{algorithm1}}, by utilizing the MC method. 
\end{itemize}

For the sake of description, we introduce a parameter $K$, similar to the Rician factor, defined as $K = \frac{\mathbb{E}||\sum_{l=1}^L \beta_l \bm{s}(\theta_l, \phi_l, r_l)||^2}{\mathbb{E}||\bm{Z}||^2}$, which represents the power ratio between $L$ deterministic paths and the additional diffuse signal components of the fading channel shown in (\ref{channel model}). Fig. \ref{different_K} investigates the OP of the proposed model and the two benchmark models under different values of $K$. From this figure, we can observe that the channel model obtained from the proposed EM-based algorithm shows an excellent agreement with the MC simulations, which confirms its ability to accurately capture the channel characteristics and verifies the accuracy of the derived expression of OP in (\ref{OP_expression}). As the value of $K$ decreases (i.e., given the power for the deterministic components while increasing the power of the diffuse components), the gaps from the near-field scheme and the far-field scheme to the MC simulations become smaller. The reason is that the impact of estimation errors is mainly on the deterministic components of the channel. As $K$ decreases, the contribution of the deterministic components to the overall signals decreases, leading to the decrease of the corresponding OP. Nevertheless, the OP obtained from the proposed scheme is always consistent with the MC simulations, indicating its effectiveness under different values of $K$.

\linespread{1.1}
\begin{table}
    \centering
    \caption{Simulation Configuration}
    \begin{tabular}{|c|c|}
        \hline
        The number of UPA antennas $N_1 \times N_2$ & $256\times 16$  \\
        \hline
        The wavelength $\lambda$ & $10$ mm \\
        \hline
        The spacing between adjacent antennas $d$ & $\lambda/2$ \\
        \hline
        The number of channel paths $L$ & $4$ \\
        \hline
        The ratio of near-field paths to far-field paths $\gamma$ & $1$ \\
        \hline
        The distribution of $\beta$ & $\lambda/4\pi r \cdot \mathcal{CN}(0, 1)$ \\
        \hline
        The distribution of $\theta$ & $\mathcal{U} (\pi/3, 2\pi/3)$ \\
        \hline
        The distribution of $\phi$ & $\mathcal{U} (-\pi/6, \pi/6)$ \\
        \hline
        The distribution of $r$ & $\mathcal{U}(4m, r_{RD})$ \tablefootnote{$r_{RD}$ represents the commonly adopted Rayleigh distance as the boundary between far-field and near-field regions, which can be calculated as $r_{RD}= \frac{2(N_1^2+N_2^2)d^2}{\lambda}$ \cite{Rayleigh Distance}.} \\
        \hline
        The transmit power of BS $P_{\rm T}$ & $40$ dBm \\
        \hline
        The noise variance at the receiver & $-96$ dBm \\
        \hline
        The number of channel samples $S$ & $100$\\
        \hline
    \end{tabular}
    \label{sim_config}
    \vspace{-0.3cm}
\end{table}


\begin{figure}
    \vspace{-0.2cm}
        \centering
        \includegraphics[width=0.48\textwidth, height=0.24\textheight]{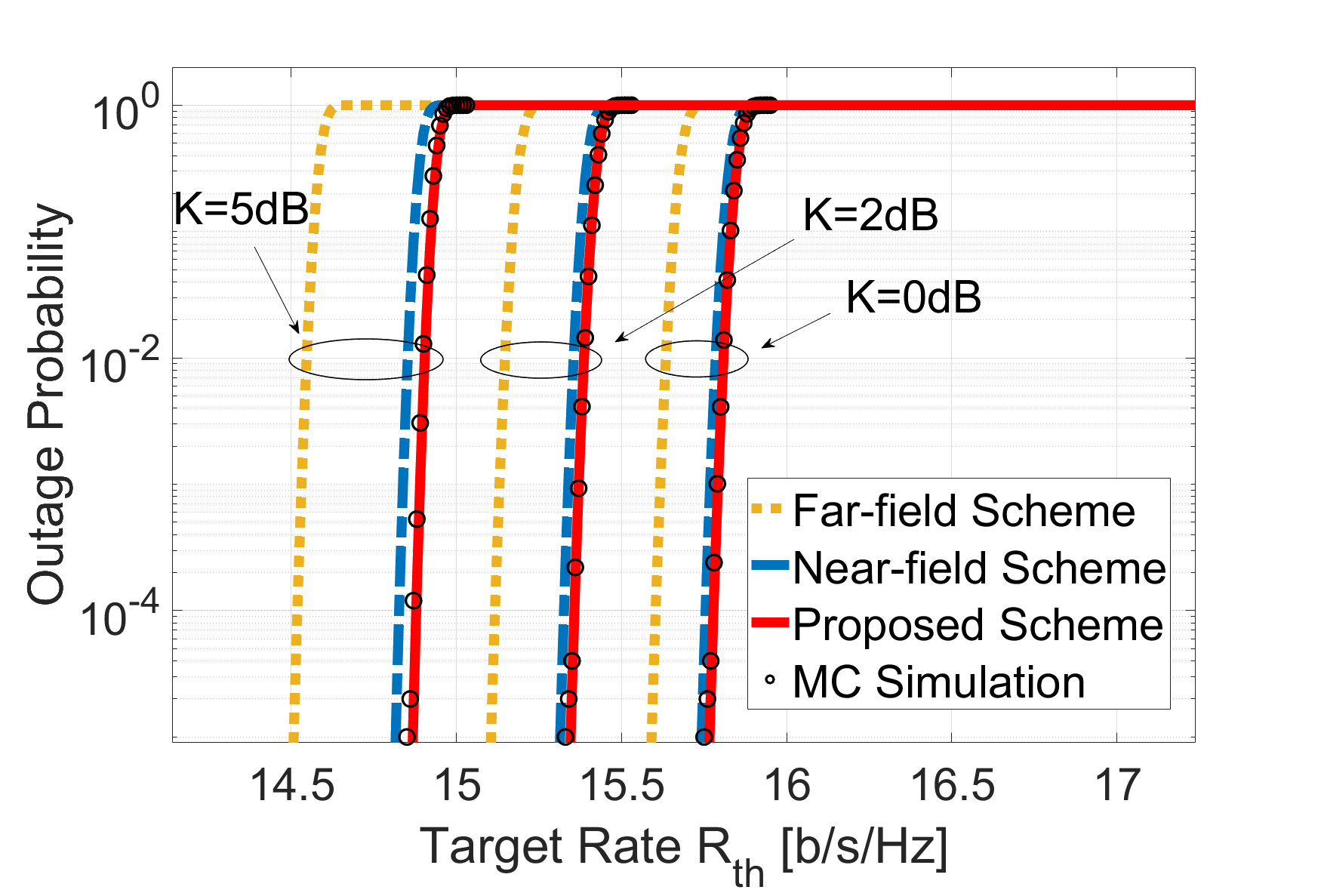}
        \caption{Outage probability performance vs. target rate $R_{\rm th}$  under different $K$, with $N_1=256$, $N_2=16$ and $\gamma=1$.}
        \label{different_K}
     \vspace{-0.1cm}
\end{figure}

\vspace{-0.1cm}
\begin{figure}
        \centering
        \includegraphics[width=0.48\textwidth]{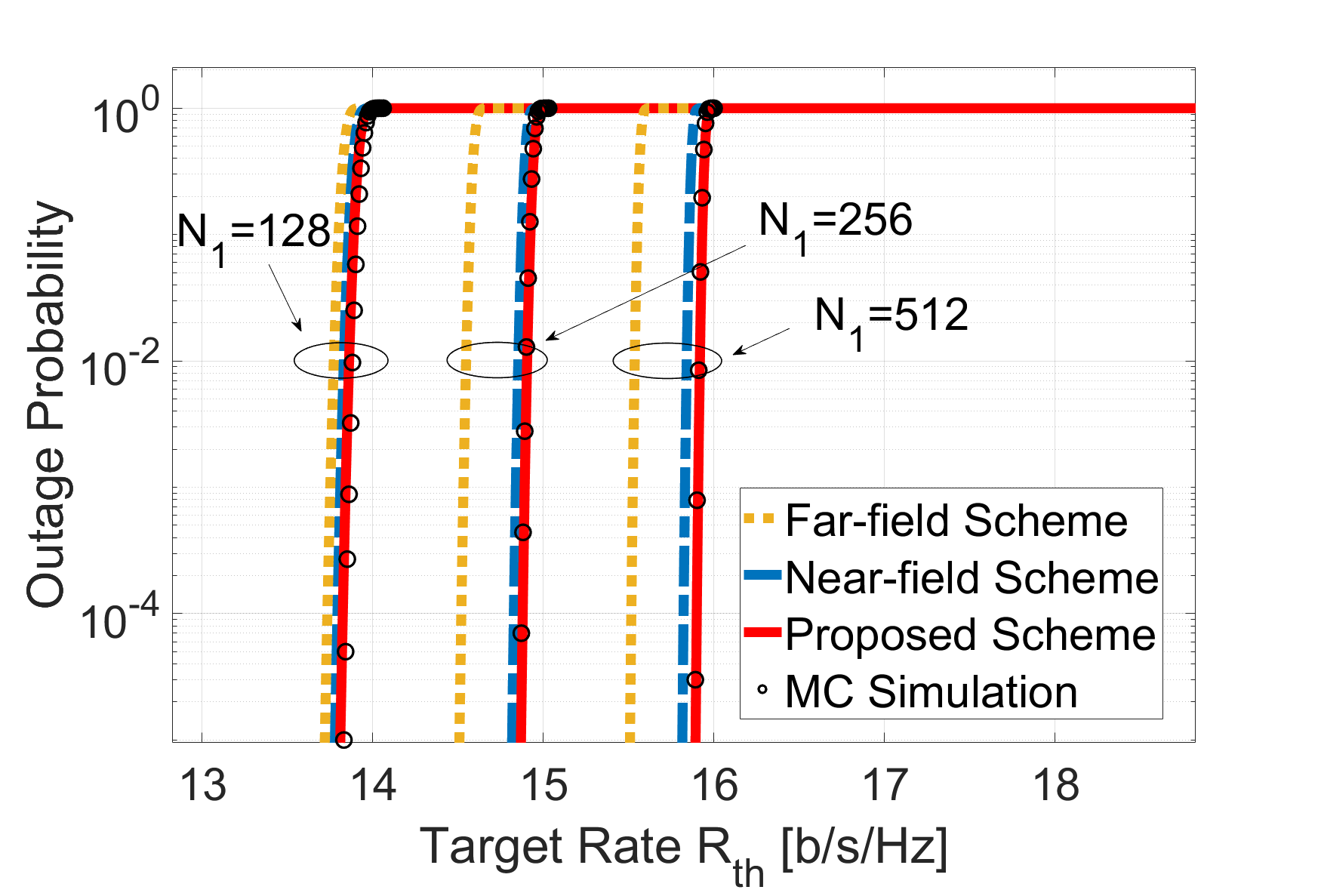}
        \caption{Outage probability performance vs. target rate $R_{\rm th}$ under different $N_1$, with $K=5$ dB and $\gamma=1$.}
        \label{different_N}
\end{figure}
\begin{figure}
        \centering
        \includegraphics[width=0.48\textwidth]{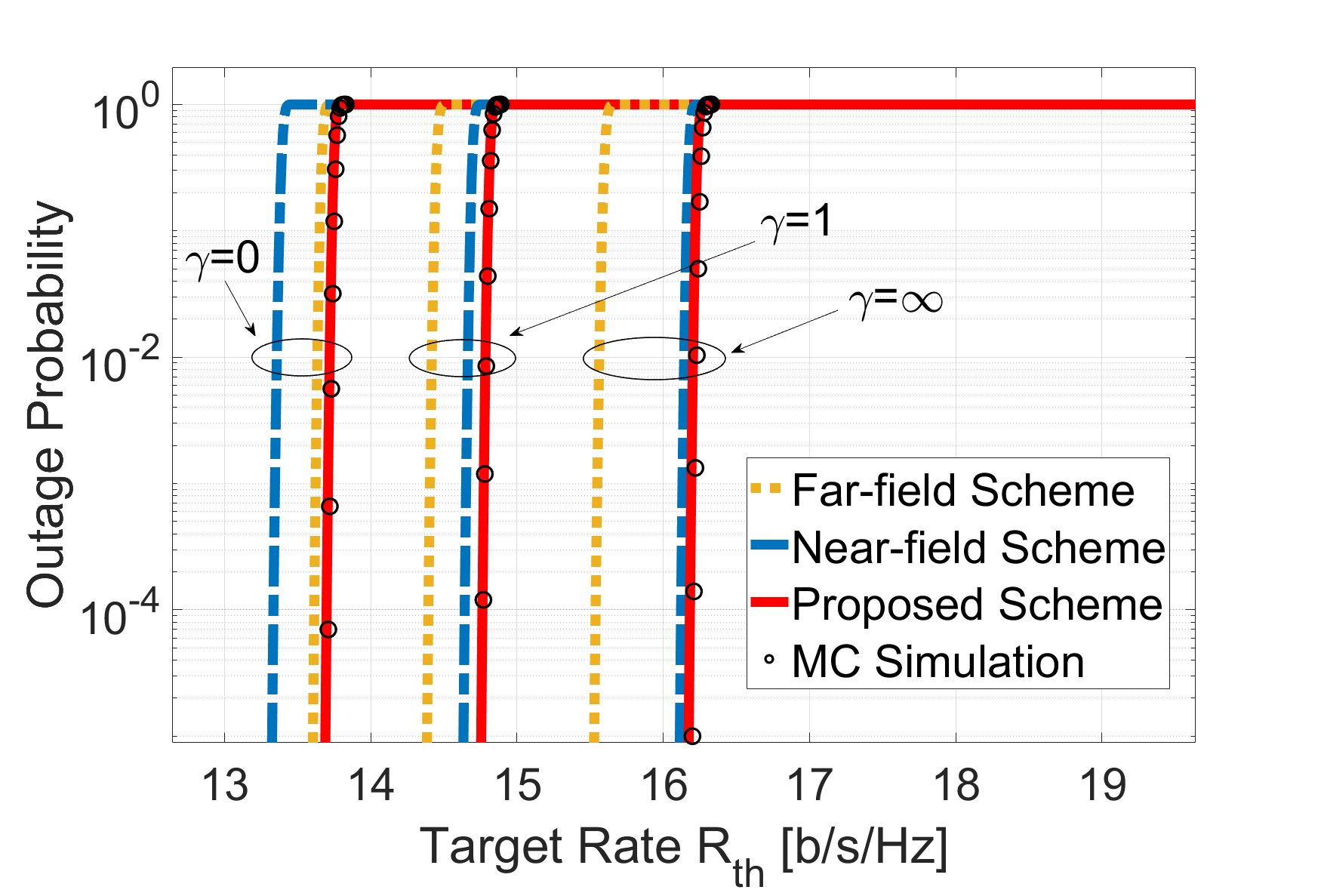}
        \caption{Outage probability performance vs. target rate $R_{\rm th}$ under different $\gamma$, with $K=5$dB, $N_1 = 256$ and $N_2 = 16$.}
        \label{different_gamma}
\end{figure}

Fig. \ref{different_N} shows the OP with different numbers of antennas, where we set $K=5$ dB and $N_2=16$. We can observe that the OP of the proposed scheme fits the MC simulations better than the near-field and far-field schemes.
However, as the number of antennas increases, the OPs of the near-field and far-field schemes show a significant degradation. This is because the near-field effect intensifies and the difference between near-field and far-field paths becomes more significant as the number of antennas increases\footnote{We assume that the near-field path is more accurately represented by the near-field steering vector, while the far-field path is more appropriately characterized by the far-field steering vector.}. Therefore, the estimation errors of the channel parameters using purely near-field or far-field algorithms become larger. The proposed scheme can adaptively determine whether each path is near-field or far-field, thus fitting the channels better.

Fig. \ref{different_gamma} shows the OP with different numbers of scatters in the near-field. Here, we use $\gamma$ to indicate the proportion of near-field paths to far-field paths, where $\gamma = 0$ means that all paths are far-field. Simulation results showed that our proposed scheme is superior to conventional near-field and far-field schemes to capture the channel environment with an arbitrary value of $\gamma$. The conventional far-field algorithm exhibits poor performance as the proportion of near-field paths in the channel increases. In contrast, we can find that the proposed model has the best agreement with the MC simulations in all three cases, which indicates that our model can also be applied in both entirely far-field and entirely near-field scenarios.

\section{Conclusion} \label{Conclusion}
In this paper, we have developed a new near-far channel model to capture the wireless propagation environment in HMIMO communications where the near field extends. By modeling whether each path is near-field or far-field as a hidden variable, we employed the EM algorithm to adaptively adjust the hidden variables and obtain the channel parameters, thereby modeling the practical near-far field channel. From the simulation results, we can obtain two conclusions: 1) The proposed near-far field model matches the MC simulation better compared to traditional near-field and far-field models; 2) The modeling accuracy still holds for different numbers of antennas and the near-far path ratio.  

\begin{appendices} 
\section{Derivation of the Gradient} \label{appendix1}
\vspace{-0.1cm}

In this appendix, the gradient of the log-likelihood function $Q(\Theta|\Theta^{j-1})$ with respect to $\theta$, $\phi$, and $r$ is provided. Firstly, the expression for $Q(\Theta|\Theta^{j-1})$ can be further simplified as

\vspace{-0.1cm}
{\setlength\abovedisplayskip{0.1cm}
    \setlength\belowdisplayskip{0.1cm}
\begin{align}
\label{simplified_Q}
    &Q(\Theta | \Theta^{j-1}) = \mathbb{E}_{\bm{z}\sim p(\bm{z}|\bm{H}, \Theta^{j-1})} \Big[\log( p(\bm{H}, \bm{z} | \Theta) )\Big] \nonumber \\
    &= \mathbb{E}_{\bm{z}\sim p(\bm{z}|\bm{H}, \Theta^{j-1})} \Big[\sum_{s=1}^S \log \Big(p(\bm{h}_s | \bm{z}, \Theta) \Big) + \log \Big( p(\bm{z} | \Theta)\Big)\Big]~.
\end{align}}

\noindent Here, $\log \Big(p(\bm{h}_s | \bm{z}, \Theta) \Big)$ is calculated according to (\ref{channel distribution}) as 
{\setlength\abovedisplayskip{0.1cm}
    \setlength\belowdisplayskip{0.1cm}
\begin{align}
\label{log_prob_h_s_to_z_Theta}
\log \Big(p(\bm{h}_s | \bm{z}, \Theta) \Big) \!&=\!
\log \bigg( \mathcal{CN}\Big(\sum_{l=1}^L \beta_l \bm{s}^H(\theta_l, \phi_l, r_l | z_l), \sigma^2\bm{I}\Big) \bigg) \nonumber\\
    \!&=\!-\bm{r}^{H}(\sigma^2\bm{I})^{-1}\bm{r} \!-\! \log \!\bigg( (\pi)^{N} |\sigma^2 I| \!\bigg),
\end{align}
where $\bm{r}=\bm{h}_s - \sum\limits_{l=1}^L \beta_l \bm{s}^H(\theta_l, \phi_l, r_l | z_l)$.} We assume that all scatters are equally likely to be distributed in the near-field and far-field regions, thus $p(\bm{z}|\Theta)$ is a constant, i.e., $p(\bm{z}|\Theta)=\frac{1}{2^L}$.

For $\theta_l$ of the $l$-th path, the gradient of $Q(\Theta|\Theta^{j-1})$ can be expressed as
{\setlength\abovedisplayskip{0.2cm}
\setlength\belowdisplayskip{0.2cm}
\begin{align}
    \frac{\partial Q(\Theta|\Theta^{j\!-\!1})}{\partial \theta_l}
    \!=\! \mathbb{E}_{\bm{z}\sim p(\bm{z}|\bm{H}, \Theta^{j-1})} \bigg\{ & \frac{\partial \Big[\sum_{s=1}^S \! \log \Big(p(\bm{h}_s | \bm{z}, \Theta) \Big) \Big]}{\partial \theta_l} \bigg\} \nonumber \\
    \!=\!\mathbb{E}_{\bm{z}\sim p(\bm{z}|\bm{H}, \Theta^{j-1})} \bigg\{ \! & \sum_{s=1}^S\frac{\partial \log \Big(p(\bm{h}_s | \bm{z}, \Theta) \Big) }{\partial \theta_l} \! \bigg\} .
\end{align}}According to (\ref{log_prob_h_s_to_z_Theta}), the gradient of $\log \!\Big(\!p(\bm{h}_s | \bm{z}, \!\Theta) \!\Big)\!$ is given by
{\setlength\abovedisplayskip{0.2cm}
    \setlength\belowdisplayskip{0.2cm}
\begin{align}
    \frac{\partial \log \Big(p(\bm{h}_s | \bm{z}, \Theta) \Big)}{\partial \theta_l} 
    &=-\frac{1}{\sigma^{2}} \frac{\partial \Big( \bm{r}^{H}\bm{r} \Big)}{\partial \theta_l} ~.
\end{align}
Since $\bm{r}=\bm{h}_s - \sum\limits_{l=1}^L \beta_l \bm{s}^H(\theta_l, \phi_l, r_l | z_l)$,} the gradient of $\bm{r}^H \bm{r}$ is given by
{\setlength\abovedisplayskip{0.2cm}
    \setlength\belowdisplayskip{0.2cm}
    \small
\begin{align}
    \label{partial_r_H_r_to_theta}
    \frac{\partial \Big( \bm{r}^{H}\bm{r} \Big)}{\partial \theta_l}
    =& (- \beta_l') \!\cdot\! \frac{\partial \bm{s}^H(\theta_l, \phi_l, r_l | z_l)}{\partial \theta_l} \!\cdot\! \bm{r} + \bm{r}^H \!\cdot\! (-\beta_l)\! \cdot \!\frac{\partial \bm{s}(\theta_l, \phi_l, r_l | z_l)}{\partial \theta_l} \nonumber \\
    =& -2 \Re \bigg\{\beta_l' \!\cdot\! \frac{\partial \bm{s}^H(\theta_l, \phi_l, r_l | z_l)}{\partial \theta_l} \!\cdot\! \bm{r} \bigg\}.
\end{align}}

\noindent By combining (\ref{simplified_Q})-(\ref{partial_r_H_r_to_theta}), we are able to derive the gradient of \(Q(\Theta|\Theta^{j-1})\) with respect to \(\theta_l\). The gradients of \(Q(\Theta|\Theta^{j-1})\) with respect to \(\phi_l\) and \(r_l\) can be obtained in a similar way. The only difference for the gradient of $\beta_l$ is that (\ref{partial_r_H_r_to_theta}) is changed to
{\setlength\abovedisplayskip{0.2cm}
    \setlength\belowdisplayskip{0.2cm}
\begin{align}
    \frac{\partial \Big( \bm{r}^{H}\bm{r} \Big)}{\partial \beta_l}
    =-\bm{r}^H \cdot \bm{s}(\theta_l, \phi_l, r_l | z_l).
\end{align}}

Finally, by stacking all gradients of $\theta_l$, $\phi_l$, $r_l$ and $\beta_l$ in a column vector, we can obtain the gradients $\nabla_\Theta Q(\Theta|\Theta^{j-1})$.

\vspace{0.1cm}
\section{Proof of Proposition \ref{proposition1}} \label{appendix2}
Here, we define $\bm{h}=[h_1, \cdots, h_N]^H$, where $h_i = X_i+jY_i$. According to (\ref{channel distribution}), $X_i \sim \mathcal{N}(\mu_{X_i}, \sigma^2/2)$, $Y_i \sim \mathcal{N}(\mu_{Y_i}, \sigma^2/2)$, where $\mu_{X_i}$ and $\mu_{Y_i}$ are the mean of $X_i$ and $Y_i$, respectively, $\mathcal{N}(\mu, \sigma^2)$ represents normal distribution with the mean $\mu$ and the variance $\sigma^2$. In this way, $||\bm{h}||^2$ can be further expressed as
{\setlength\abovedisplayskip{0.1cm}
    \setlength\belowdisplayskip{0.1cm}
    \begin{align}
        \label{h_2}
        ||\bm{h}||^2 = \sum_{i=1}^{N} \Big(X_i^2 + Y_i^2 \Big).
    \end{align}
    Since $X_i$ and $Y_i$ have the same variance $\sigma^2$,} $\frac{X_i}{\sqrt{\sigma^2/2}}$ and $\frac{Y_i}{\sqrt{\sigma^2/2}}$ satisfy 
{\setlength\abovedisplayskip{0.2cm}
    \setlength\belowdisplayskip{0.1cm}
    \begin{align}
        \frac{X_i}{\sqrt{\sigma^2/2}} \sim \mathcal{N} \Big(\frac{\mu_{X_i}}{\sqrt{\sigma^2/2}}, 1 \Big),
        \frac{Y_i}{\sqrt{\sigma^2/2}} \sim \mathcal{N} \Big(\frac{\mu_{Y_i}}{\sqrt{\sigma^2/2}}, 1 \Big).
    \end{align}
    According to (\ref{h_2}),}
    {\setlength\abovedisplayskip{0.2cm}
    \setlength\belowdisplayskip{0.1cm}
    \begin{align}
        \frac{2}{\sigma^2}||\bm{h}||^2 = \sum_{i=1}^{N}\Big(\frac{X_i^2}{\sigma^2/2} + \frac{Y_i^2}{\sigma^2/2}\Big),
    \end{align}
    i.e.,} $\frac{2}{\sigma^2}||\bm{h}||^2$ is the sum of squares of $2N$ independent, normally distributed random variables with different means and unit variances. Therefore, $\frac{2}{\sigma^2}||\bm{h}||^2$ follows the non-central chi-squared distribution~\cite{noncentral_chi}. I.e., 
    {\setlength\abovedisplayskip{0.1cm}
    \setlength\belowdisplayskip{0.1cm}
    \begin{align}
        \frac{2}{\sigma^2}||\bm{h}||^2 \sim \chi_k'^2(\lambda),
    \end{align}
    where $k$ is the number of degrees of freedom,} and $\lambda$ is the non-centrality parameter. The number of degrees of freedom is equal to the number of independent, normally distributed random variables, i.e., $k=2N$. The non-centrality parameter $\lambda$ is equal to the sum of squares of the means of the random variables, which can be expressed as
    {\setlength\abovedisplayskip{0.2cm}
    \setlength\belowdisplayskip{0.2cm}
    \begin{align}
        \lambda &= \sum_{i=1}^N \Big(\frac{\mu_{X_i}^2}{\sigma^2/2} + \frac{\mu_{Y_i}^2}{\sigma^2/2} \Big).
    \end{align}
    According to (\ref{channel distribution}), }
     {\setlength\abovedisplayskip{0.1cm}
    \setlength\belowdisplayskip{0.1cm}
    \begin{align}
    \sum_{i=1}^N \Big(\mu_{X_i}^2 + \mu_{Y_i}^2 \Big) = \left \Vert \sum_{l=1}^L \beta_l \bm{s}(\theta_l, \phi_l, r_l) \right \Vert^2 .
    \end{align}
    Therefore,} $\lambda$ can be calculated as
{\setlength\abovedisplayskip{0.1cm}
    \setlength\belowdisplayskip{0.1cm}
    \begin{align}
        \lambda = \frac{\left\Vert\sum_{l=1}^L \beta_l \bm{s}(\theta_l, \phi_l, r_l)\right\Vert^2}{\sigma^2/2},
    \end{align}
    and $\frac{2}{\sigma^2}||\bm{h}||^2$ follows}
{\setlength\abovedisplayskip{0.1cm}
    \setlength\belowdisplayskip{0.2cm}
    \begin{align}
        \frac{2}{\sigma^2}||\bm{h}||^2 \sim \chi_{2N}'^2 \bigg(\frac{\left\Vert\sum_{l=1}^L \beta_l \bm{s}(\theta_l, \phi_l, r_l) \right\Vert^2}{\sigma^2/2} \bigg).
    \end{align}
    The expression of OP in (\ref{OP_expression}) can be reformulated as}
    {\setlength\abovedisplayskip{0.2cm}
    \setlength\belowdisplayskip{0.2cm}
    \begin{align}
        OP &=  Pr\Big\{ \log_2(1+\frac{P_{\rm T}}{\hat{\sigma}^2}\bm{h}^H\bm{h}) < R_{\rm th} \Big\} \nonumber \\ 
        &= Pr\Big\{ \frac{2}{\sigma^2}||\bm{h}||^2 < \frac{(2^{R_{\rm th}} - 1)\hat{\sigma}^2 }{P_{\rm T} \cdot \sigma^2/2} \Big\} \nonumber \\
        &= F\bigg(\frac{(2^{R_{\rm th}} - 1)\hat{\sigma}^2 }{P_{\rm T} \cdot \sigma^2/2}; 2N, \frac{ \left\Vert \sum_{l=1}^L \beta_l \bm{s}(\theta_l, \phi_l, r_l)\right\Vert^2}{\sigma^2/2} \bigg),
    \end{align}
    where $F(x; k, \lambda)$ denotes the CDF of the non-central chi-squared distribution $\chi_k'^2(\lambda)$.}

\end{appendices}

\vspace{0.2cm}
\addtolength{\textheight}{-0.01in}


\begin{thebibliography}{0}
    \vspace{0.2cm}
    \bibitem{HMIMO_review}
    S. Zeng, H. Zhang, B. Di, H. Qin, X. Su, and L. Song, “Reconfigurable Refractive Surfaces: An Energy-Efficient Way to Holographic MIMO,” \emph{IEEE Commun. Lett.}, vol. 26, no. 10, pp. 2490–2494, Oct. 2022.

    \bibitem{HMIMO_integrate}
    H. Zhang and B. Di, “Intelligent Omni-Surfaces: Simultaneous Refraction and Reflection for Full-Dimensional Wireless Communications,” \emph{IEEE Commun. Surveys Tuts.}, vol. 24, no. 4, pp. 1997–2028, 2022.

    \bibitem{OMP_2016}
    J. Lee, G. -T. Gil and Y. H. Lee, "Channel Estimation via Orthogonal Matching Pursuit for Hybrid MIMO Systems in Millimeter Wave Communications," \emph{IEEE Trans. Commun.}, vol. 64, no. 6, pp. 2370-2386, Jun. 2016.

    
    \bibitem{Cui_2022}
    M. Cui and L. Dai, “Channel Estimation for Extremely Large-Scale MIMO: Far-Field or Near-Field?” \emph{IEEE Trans. Commun.}, vol. 70, no. 4, pp. 2663–2677, Apr. 2022.

    

    \bibitem{Wei_2022}
    X. Wei and L. Dai, “Channel Estimation for Extremely Large-Scale Massive MIMO: Far-Field, Near-Field, or Hybrid-Field?” \emph{IEEE Commun. Lett.}, vol. 26, no. 1, pp. 177–181, Jan. 2022.

    \bibitem{Hu_2023}
    Z. Hu, C. Chen, Y. Jin, L. Zhou, and Q. Wei, “Hybrid-Field Channel Estimation for Extremely Large-Scale Massive MIMO System,” \emph{IEEE Commun. Lett.}, vol. 27, no. 1, pp. 303–307, Jan. 2023.

    \bibitem{Rayleigh Distance}
    Y. Zhang, B. Di, H. Zhang, and L. Song, “Near-Far Field Beamforming for Holographic Multiple-Input Multiple-Output,” \emph{J. Commun. Inf. Netw.}, vol. 8, no. 2, pp. 99–110, Jun. 2023.


    \bibitem{J.D._2020}
    J. D. V. Sánchez, D. P. M. Osorio, F. J. López-Martínez, M. C. P. Paredes and L. F. Urquiza-Aguiar, "On the Secrecy Performance over N-Wave with Diffuse Power Fading Channel," \emph{IEEE Trans. Veh. Technol.}, vol. 69, no. 12, pp. 15137-15148, Dec. 2020.

    \bibitem{channel_model_2022}
    S. Jiang, W. Wang, Y. Miao, W. Fan, and A. F. Molisch, “A Survey of Dense Multipath and Its Impact on Wireless Systems,” \emph{IEEE Open J. Antennas Propag.}, vol. 3, pp. 435–460, 2022.


    \bibitem{near_field_vector}
    Z. Wu and L. Dai, "Multiple Access for Near-Field Communications: SDMA or LDMA?" \emph{IEEE J. Sel. Areas Commun.}, vol. 41, no. 6, pp. 1918-1935, Jun. 2023.


    \bibitem{EM_algorithm_paper}
    A. P. Dempster, N. M. Laird, and D. B. Rubin, “Maximum Likelihood from Incomplete Data via the EM Algorithm,” \emph{J. R. Statist. Soc. B}, vol. 39, no. 1, pp. 1–22, 1977.


    \bibitem{Armijo}
    Larry Armijo, “Minimization of Functions Having Lipschitz Continuous First Partial Derivatives,” \emph{Pac. J. Math.}, vol. 16, no. 1, pp. 1–3, Jan. 1966.


    \bibitem{Zeng_2021}
    S. Zeng, H. Zhang, B. Di, Z. Han, and L. Song, “Reconfigurable Intelligent Surface (RIS) Assisted Wireless Coverage Extension: RIS Orientation and Location Optimization,” \emph{IEEE Commun. Lett.}, vol. 25, no. 1, pp. 269–273, Jan. 2021.
    
    \bibitem{noncentral_chi}
    P. B. Patnaik, “The Non-Central $\chi^2$- and F-Distribution and Their Applications,” \emph{Biometrika}, vol. 36, no. 1/2, pp. 202–232, 1949.

\end{thebibliography}
\end{document}